\documentclass[sigconf,screen=true]{acmart}

\setcopyright{rightsretained}
\acmPrice{}
\acmDOI{10.1145/3368089.3409732}
\acmYear{2020}
\copyrightyear{2020}
\acmSubmissionID{fse20main-p447-p}
\acmISBN{978-1-4503-7043-1/20/11}
\acmConference[ESEC/FSE '20]{Proceedings of the 28th ACM Joint European Software Engineering Conference and Symposium on the Foundations of Software Engineering}{November 8--13, 2020}{Virtual Event, USA}
\acmBooktitle{Proceedings of the 28th ACM Joint European Software Engineering Conference and Symposium on the Foundations of Software Engineering (ESEC/FSE '20), November 8--13, 2020, Virtual Event, USA}

\usepackage{booktabs}   
\usepackage{subcaption} 

\usepackage{amsmath, listings, amsthm, proof, xspace}
\usepackage{mathtools}
\usepackage{stmaryrd}
\usepackage{dsfont}
\usepackage{flushend}
\usepackage{microtype}

\usepackage{enumitem}

\usepackage{multirow}
\usepackage[ruled,vlined]{algorithm2e}
\usepackage{tikz}
\usetikzlibrary{arrows,automata,positioning, shapes, shapes.geometric, fit, calc}

\usepackage{amsthm}

\newtheorem{theorem}{Theorem}

\newtheorem{definition}{Definition}
\newtheorem{example}{Example}

\newcommand{\alt}{~~|~~}

\newcommand{\binopdef}     \oplus 
\newcommand{\unopdef}      \ominus 

\makeatletter
\DeclareRobustCommand*\cal{\@fontswitch\relax\mathcal}
\makeatother

\def\cA{{\cal A}}

\def\cD{{\cal D}}

\def\cL{{\cal L}}

\newcommand{\tup} [1] {\langle #1 \rangle}

    \newcommand{\infral} [3] {\infer[\textsc{#3}]{\begin{array}{c} #2 \end{array} }{ \begin{array}{c} #1  \end{array} } }

\newcommand{\Comment} [1] {}

\newcommand{\arxiv} [1] {}

\newcommand{\dref} [1] {}

\usepackage{array}
\newcolumntype{H}{>{\setbox0=\hbox\bgroup}c<{\egroup}@{}}

\begin{document}

\title[Inductive Program Synthesis over Noisy Data]{Inductive Program Synthesis over Noisy Data}         


\author{Shivam Handa}
\affiliation{
  \department{Electrical Engineering and Computer Science}              
  \institution{Massachusetts Institute of Technology}            
  \country{U.S.A.}                    
}
\email{shivam@mit.edu}          

\author{Martin C. Rinard}
\affiliation{
  \department{Electrical Engineering and Computer Science}             
  \institution{Massachusetts Institute of Technology}           
  \country{U.S.A.}                   
}
\email{rinard@csail.mit.edu}         

\begin{abstract}
    We present a new framework and associated synthesis algorithms for
program synthesis over noisy data, i.e., data that may contain
incorrect/corrupted input-output examples.  This framework is based on
an extension of finite tree automata called {\em state-weighted finite tree
automata}.  We show how to apply this framework to formulate and solve
a variety of program synthesis problems over noisy data.  Results from
our implemented system running on problems from the SyGuS 2018 benchmark
suite highlight its ability to successfully synthesize programs in the
face of noisy data sets, including the ability to synthesize a correct
program even when every input-output example in the data set is corrupted.

\end{abstract}



\maketitle

\vspace{-.1in}
\section{Introduction}
In recent years there has been significant interest in learning programs
from input-output examples. 
These techniques have been successfully used to synthesize programs
for domains such as string and format transformations~\cite{gulwani2011automating, 
singh2016transforming}, 
data wrangling~\cite{feng2017component}, data completion~\cite{wang2017synthesis}, and
data structure manipulation~\cite{feser2015synthesizing, 
osera2015type, yaghmazadeh2016synthesizing}.
Even though these efforts have been largely successful, they 
do not aspire to work with noisy data sets that 
may contain corrupted input-output examples.

We present a new program synthesis technique that is designed to work
with noisy/corrupted data sets. Given: 
\begin{itemize}[leftmargin=*]
\item {\bf Programs:} A collection of programs $p$ defined by a 
grammar $G$ and a bounded scope threshold $d$, 
\item {\bf Data Set:} A data set $\cD = \{ (\sigma_1, o_1 ), \ldots, ( \sigma_n,
    o_n ) \}$ of 
input-output examples, 
\item {\bf Loss Function:} A loss function $\cL(p,\cD)$ that measures the cost of the input-output examples on which
$p$ produces a different output than the output in the data set $\cD$, 
\item {\bf Complexity Measure:} A complexity measure $C(p)$ 
    that measures the complexity
of a given program $p$, 
\item {\bf Objective Function:} An arbitrary objective function
    $U(l, c)$, which maps loss $l$ and complexity $c$ to a totally ordered set, 
        such that for all values of $l$, $U(l, c)$ is 
        monotonically nondecreasing with respect to $c$, 
\end{itemize}
our program synthesis technique produces a program $p$ that 
minimizes $U(\cL(p,\cD),C(p))$.
Example problems that our program synthesis technique can solve include:

\begin{itemize}[leftmargin=*]

\item {\bf Best Fit Program Synthesis:} Given a potentially noisy data set $\cD$, find a 
{\em best fit} program $p$ for $\cD$, i.e., a $p$ that 
minimizes $\cL(p,\cD)$. 

\item {\bf Accuracy vs. Complexity Tradeoff:} Given a data set $\cD$, 
    find $p$ that minimizes the weighted sum $\cL(p,\cD) + \lambda \cdot C(p)$. 
    This problem enables the programmer to define and
    minimize a weighted tradeoff between the complexity of the program
    and the loss. 

\item {\bf Data Cleaning and Correction:} Given a data set $\cD$, 
find $p$ that minimizes the loss $\cL(p,\cD)$. Input-output
examples with nonzero loss are identified as corrupted and either 
1) filtered out or 2) replaced with the output from the synthesized
program. 

\item{\bf Bayesian Program Synthesis:} Given a data set $\cD$ and a probability
    distribution $\pi(p)$ over programs $p$, find the most probable program $p$ given $\cD$. 


\item {\bf Best Program for Given Accuracy:} Given a data set $\cD$ and a bound $b$, find $p$ that minimizes $C(p)$
    subject to $\cL(p,\cD) \leq b$. One use case finds the simplest program that agrees with the data set $\cD$ on 
    at least $n-b$ input-output examples. 

\item {\bf Forced Accuracy:} Given data sets $\cD'$, $\cD$, where $\cD' \subseteq \cD$, find $p$ that minimizes 
    the weighted sum $\cL(p,\cD) + \lambda \cdot C(p)$ subject to $\cL(p,\cD') \leq b$. 
    One use case finds a program $p$ which minimizes the loss over the data set $\cD$ but is 
    always correct for $\cD'$. 

\item {\bf Approximate Program Synthesis:} Given a clean (noise-free) data set $\cD$,
find the least complex program $p$ that minimizes the loss $\cL(p,\cD)$. Here the
goal is not to work with a noisy data set, but instead 
to find the best approximate solution to a synthesis problem when
an exact solution does not exist within the collection of considered programs $p$. 


\Comment{
Program synthesis systems are typically designed 
to synthesize a program that matches all input-output examples in the data set. Our technique, in
contrast, supports {\em state-weighted program synthesis}: given a clean (noise-free) data set $\cD$, 
find the least complex program $p$ that minimizes the loss $\cL(p,\cD)$ over the data set. 

        Our technique therefore enables the approximate solution of inductive program synthesis problems when 
an exact solution does not exist within the DSL or within a finite bounded program search scope. 
}

\end{itemize}

\vspace{-.1in}
\subsection{Noise Models and Discrete Noise Sources}

We work with noise models that assume a (hidden) clean data set 
combined with a noise source that delivers the noisy data set presented to the 
program synthesis system. Like many inductive program synthesis systems~\cite{gulwani2011automating, 
singh2016transforming}, one
target is {\em discrete problems} that involve 
discrete data such as strings, data structures, or tablular data. 
In contrast to traditional machine learning problems, which often involve
continuous noise sources~\cite{bishop2006pattern}, the noise sources for 
discrete problems often involve {\em discrete noise} ---
noise that involves a discrete change that affects only part of each output,
leaving the remaining parts intact and uncorrupted. 

\vspace{-.1in}
\subsection{Loss Functions and Use Cases}

Different loss functions can be appropriate for different noise sources and
use cases. The $0/1$ loss function, which counts the number of input-output
examples where the data set $\cD$ and synthesized program $p$ do not agree, 
is a general loss function that can be appropriate when the focus is to maximize the 
number of inputs for which the synthesized program $p$ produces the correct output. 
The Damerau-Levenshtein (DL) loss function~\cite{damerau1964technique}, which measures the edit difference
under character insertions, deletions, substitutions, and/or transpositions, 
extracts information present in partially corrupted outputs and can be appropriate for measuring
discrete noise in input-output examples involving text strings. 
The $0/\infty$ loss function, which is $\infty$ unless $p$ agrees
with all of the input-output examples in the data set $\cD$, 
specializes our technique to the standard program synthesis scenario that 
requires the synthesized program to agree with all input-output examples.  

Because discrete noise sources often leave parts of corrupted outputs intact, 
{\em exact program synthesis} (i.e., synthesizing a program that agrees with all outputs in the hidden 
clean data set) is often possible even when all outputs in the data set are corrupted. 
Our experimental results (Section~\ref{sec:results}) indicate that matching the loss
function to the characteristics of the discrete noise source can enable very
accurate program synthesis even when 1) there are only a handful of input-output examples in 
the data and 2) all of the outputs in the data set are corrupted. We attribute this
success to the ability of our synthesis technique, working in conjunction with
an appropriately designed loss function, to effectively extract information from outputs 
corrupted by discrete noise sources. 

\subsection{Technique} 

Our technique augments finite tree automata (FTA) to associate accepting states with weights
that capture the loss for the output associated with each accepting state.  Given a data set $\cD$, 
the resulting {\em state-weighted finite tree automata} (SFTA) partition the programs $p$ defined
by the grammar $G$ into equivalence classes. Each equivalence class consists
of all programs with identical input-output behavior on the inputs in $\cD$. 
All programs in a given equivalence class therefore have the same loss over $\cD$. 
The technique then uses dynamic programming to find the minimum complexity program $p$ in each 
equivalence class~\cite{gallo1993directed}. From this set of minimum complexity programs, the technique
then finds the program $p$ that minimizes the objective function $U(p,\cD)$. 

\vspace{-.1in}
\subsection{Experimental Results} 
We have implemented our technique and applied it to various programs in the 
SyGuS 2018 benchmark set~\cite{SyGuS2018}.
The results indicate that the technique is effective at solving program synthesis
problems over strings with modestly sized solutions even in the presence of 
substantial noise. For discrete noise sources and a loss function that is a 
good match for the noise source, the technique is typically able to extract
enough information left intact in corrupted outputs to synthesize a correct program
even when all outputs are corrupted (in this paper we consider a synthesized program to be correct
if it agrees with all input-output examples in the original hidden clean data set). 
Even with the 0/1 loss function, which does not aspire to extract any information from
corrupted outputs, the technique is typically able to synthesize a correct program
even with only a few correct (uncorrupted) input-output examples in the data set. 
Overall the results highlight the potential for effective program synthesis even
in the presence of substantial noise. 

\vspace{-.1in}
\subsection{Contributions}

This paper makes the following contributions:
\begin{itemize}[leftmargin=*]
\item {\bf Technique:} It presents an implemented technique for inductive program
synthesis over noisy data. The technique uses an extension of finite
tree automata, {\em state-weighted finite tree automata}, to synthesize
programs that minimize an objective function involving the loss over
the input data set and the complexity of the synthesized program. 

\item {\bf Use Cases:} It presents multiple uses cases including
best fit program synthesis for noisy data sets, 
navigating accuracy vs. complexity tradeoffs, Bayesian program synthesis, 
identifying and correcting corrupted data, and 
approximate program synthesis.

\item {\bf Experimental Results:} It presents experimental results
from our implemented system on the SyGuS 2018 benchmark set. These results 
characterize the scalability of the technique and highlight interactions between
the DSL, the noise source, the loss function, and the overall effectiveness of the 
synthesis technique. In particular, they highlight the ability of the 
technique to, given a close match between the noise source and the loss
function, synthesize a correct program $p$ even when 1) there are
only a handful of input-output examples in the data set $\cD$ and 2) 
all outputs are corrupted. 

\end{itemize}

\section{Preliminaries}
We next review finite tree automata (FTA) and 
FTA-based inductive program synthesis.

\subsection{Finite Tree Automata}

{\it Finite tree automata} are a type of state machine which accept 
trees rather than strings. They generalize standard finite automata to 
describe a regular language over trees.

\begin{definition}[\bf FTA]
    A (bottom-up) finite tree automaton (FTA) over alphabet $F$ is a tuple
    $\cA = (Q, F, Q_f, \Delta)$ where $Q$ is a set of states, $Q_f
    \subseteq Q$ is the set of accepting states and $\Delta$ is a set of
    transitions of the form $f(q_1, \ldots, q_k) \rightarrow q$ where 
    $q, q_1, \ldots q_k$ are states, $f \in F$.
\end{definition}

Every symbol $f$ in alphabet $F$ has an associated arity. The set 
$F_k \subseteq F$ is the set of all $k$-arity symbols in $F$.
$0$-arity terms $t$ in $F$ are viewed as single node trees (leaves of trees).
$t$ is accepted by an FTA if we can rewrite $t$ to some state $q \in Q_f$ using
rules in $\Delta$. The language of an $FTA$ $\cA$, denoted
by $\cL(\cA)$, corresponds to the set of all
ground terms accepted by $\cA$.

\begin{example}
    Consider the tree automaton $\cA$ defined by states $Q = \{q_{T},
    q_{F}\}$, $F_0 = \{\mathsf{True}, \mathsf{False}\}$, 
    $F_1 = \mathsf{not}$, $F_2 = \{\mathsf{and}\}$, 
    final states $Q_f = \{q_{T}\}$ and the following transition rules $\Delta$:
    \[
        \begin{array}{cc}
            \mathsf{True} \rightarrow q_T 
            &
            \mathsf{False} \rightarrow q_F
            \\
            \mathsf{not}(q_T) \rightarrow q_F
            &
            \mathsf{not}(q_F) \rightarrow q_T
            \\
            \mathsf{and}(q_T, q_T) \rightarrow q_T
            &
            \mathsf{and}(q_F, q_T) \rightarrow q_F
            \\
            \mathsf{and}(q_T, q_F) \rightarrow q_F
            &
            \mathsf{and}(q_F, q_F) \rightarrow q_F
            \\
            \mathsf{or}(q_T, q_T) \rightarrow q_T
            &
            \mathsf{or}(q_F, q_T) \rightarrow q_T
            \\
            \mathsf{or}(q_T, q_F) \rightarrow q_T
            &
            \mathsf{or}(q_F, q_F) \rightarrow q_F
            \\
        \end{array}
    \]
\end{example}

The above tree automaton accepts all propositional logic formulas
over $\mathsf{True}$ and $\mathsf{False}$
which evaluate to $\mathsf{True}$. 
Figure~\ref{fig:formula_tree} presents the tree for the formula
$\mathsf{and}(\mathsf{True}, \mathsf{not}(\mathsf{False}))$. 

\begin{figure}
    \begin{tikzpicture}[shorten >=1pt,node distance=1.25cm,on grid]
        \node[state]   (q1)                {$\mathsf{and}$};
        \node[state]           (q2) [above right=of q1] {$\mathsf{True}$};
        \node[state] (q3) [below right=of q1] {$\mathsf{not}$};
        \node[state] (q4) [right=of q3] {$\mathsf{False}$};
        \path[->] (q2) edge                node [] {} (q1)
                    (q3) edge node [] {} (q1)
                    (q4) edge node [] {} (q3);
\end{tikzpicture}
\vspace{-.1in}
    \caption{Tree for formula 
\vspace{-.1in}
    $\mathsf{and}(\mathsf{True}, \mathsf{not}(\mathsf{False}))$}
    \label{fig:formula_tree}
\end{figure}

\subsection{Domain Specific Languages (DSLs)}
\label{sec:dsl}

We next define the programs we consider, how inputs to the program are 
specified, and the program semantics. Without loss of generality, we assume
programs $p$ are specified as parse trees in a domain-specific language (DSL) grammar $G$. 
Internal nodes represent function invocations; leaves are constants/0-arity symbols in the DSL. 
A program $p$ executes in an input $\sigma$.  $\llbracket p \rrbracket \sigma$ denotes the
output of $p$ on input $\sigma$ ($\llbracket . \rrbracket$ is defined in Figure~\ref{fig:exec_sem}).

\begin{figure}
    \[
        \begin{array}{c}
            \begin{array}{cc}
            \infral{}
            {\llbracket c \rrbracket \sigma \Rightarrow c}
            {(Constant)}
            &
            \infral{}
            {\llbracket x \rrbracket \sigma \Rightarrow \sigma(x)}
            {(Variable)}
            \end{array}
            \\
            \\
            \infral{
            \llbracket n_1 \rrbracket \sigma \Rightarrow v_1 ~~~~~ 
            \llbracket n_2 \rrbracket \sigma \Rightarrow v_2 ~~~~ \ldots ~~~~
            \llbracket n_k \rrbracket \sigma \Rightarrow v_k
            }
            {\llbracket f(n_1, n_2, \ldots n_k) \rrbracket \sigma  \Rightarrow
            f(v_1, v_2, \ldots v_k)}{(Function)}
    \end{array}
    \]
\vspace{-.1in}
    \caption{Execution semantics for program $p$}
\vspace{-.1in}
    \label{fig:exec_sem}
\end{figure}

All valid programs (which can be executed) are defined by a DSL grammar $G = (T, N, P, s_0)$ where:
\begin{itemize}[leftmargin=*]
    \item $T$ is a set of terminal symbols. These may include 
        constants and symbols which may change value depending on the input
        $\sigma$.
    \item $N$ is the set of nonterminals that represent subexpressions in our
        DSL.
    \item $P$ is the set of` production rules of the form \\ $s \rightarrow f(s_1,
        \ldots, s_n)$, where $f$ is a built-in function in the DSL and 
        $s, s_1, \ldots, s_n$ are non-terminals in the grammar.
    \item $s_0 \in N$ is the start non-terminal in the grammar.
\end{itemize}

We assume that we are given a black box implementation of each built-in function $f$
in the DSL.  In general, all techniques explored within this paper can be generalized to
any DSL which can be specified within the above framework.

\begin{example}\label{ex:dsl}
    The following DSL defines expressions over input x, 
    constants 2 and 3, and addition and multiplication:
    \[
        \begin{array}{rcl}
            n &:=& x \alt n + t \alt n \times t; \\
            t &:=& 2 \alt 3;
        \end{array}
    \]
\end{example}

\subsection{Concrete Finite Tree Automata}\label{subsec:cfta}
We review the approach introduced by \cite{wang2017program, wang2017synthesis} to use 
finite tree automata to solve synthesis tasks over a broad class of DSLs.
Given a DSL and a set of input-output examples, a {\it Concrete Finite Tree
Automaton} {\bf (CFTA)} is a tree automaton which accepts all
trees representing DSL programs consistent with the input-output examples
and nothing else. The states of the FTA correspond to concrete values
and the transitions are obtained using the semantics of the DSL constructs.

\begin{figure}[t]
    \[
        \begin{array}{c}
        \begin{array}{cc}
            \infral{t \in T, ~~~~~ \llbracket t \rrbracket \sigma = c}
            {q^c_t \in Q}
            {(Term)}
            &
            \infral{q^{o}_{s_0} \in Q}
            {q^{o}_{s_0} \in Q_f}
            {(Final)}
            \\
         \end{array}
            \\
            \\
            \infral{s \rightarrow f(s_1, \ldots s_k) \in P, ~~~~~~~~~ 
            \{ q_{s_1}^{c_1}, \ldots, q_{s_k}^{c_k} \} \subseteq Q,\\
            \llbracket f(c_1, \ldots c_k) \rrbracket\sigma = c}
            {q_s^c \in Q, ~~~~~ f(q_{s_1}^{c_1}, \ldots, q_{s_k}^{c_k}) 
            \rightarrow q_s^c \in \Delta}{(Prod)}
        \end{array}
    \]
\vspace{-.1in}
    \caption{Rules for constructing a CFTA $\cA = (Q, F, Q_f, \Delta)$
    given input $\sigma$, output $o$, and grammar $G = (T, N, P, s_0)$.}
\vspace{-.1in}
    \label{fig:cfta_rules}
\end{figure}

Given an input-output example $(\sigma, o)$ and DSL $(G, \llbracket . \rrbracket)$, 
construct a CFTA using the rules in Figure~\ref{fig:cfta_rules}. The alphabet 
of the CFTA contains built-in functions within the DSL. The states 
in the CFTA are of the form $q_s^c$, where $s$ is a symbol (terminal or
non-terminal) in $G$ and $c$ is a concrete value.
The existence of state $q^c_s$ implies that there exists a partial program  
which can 
map $\sigma$ to concrete value $c$. Similarly, the existence of 
transition $f(q_{s_1}^{c_1}, q_{s_2}^{c_2} \ldots q_{s_k}^{c_k}) \rightarrow
q_s^c$ implies $f(c_1, c_2 \ldots c_k) = c$. 

The $\mathsf{Term}$ rule states that if we have a terminal $t$ (either 
a constant in our language or input symbol $x$), execute it with the input
$\sigma$ and construct a state $q_t^c$ (where $c = \llbracket t
\rrbracket \sigma$). The $\mathsf{Final}$ rule states
that, given start symbol $s_0$ and we expect $o$ as the output, if
$q^o_{s_0}$ exists, then we have an accepting state. The $\mathsf{Prod}$  rule 
states that, if we have a production rule 
$f(s_1, s_2, \ldots s_k) \rightarrow s \in \Delta$, and 
there exists states $q_{s_1}^{c_1}, q_{s_2}^{c_2} \ldots q_{s_k}^{c_k} \in Q$, 
then we also have state $q_s^c$ in the CFTA and a transition 
$f(q_{s_1}^{c_1}, q_{s_2}^{c_2}, \ldots q_{s_k}^{c_k}) \rightarrow
q_{s}^{c}$. 

The language of the CFTA constructed from Figure~\ref{fig:cfta_rules} is exactly 
the set of parse trees of DSL programs that are consistent
with our input-output example (i.e., maps input $\sigma$ to output
$o$).

In general, the rules in Figure~\ref{fig:cfta_rules} may result in a CFTA which 
has infinitely many states. To control the size of the resulting CFTA, we do not 
add a new state within the constructed CFTA if the smallest tree it will accept is larger 
than a given threshold $d$. 
This results in a CFTA which accepts all programs which are consistent with the 
input-output example but are smaller than the given threshold (it may
accept some programs which are larger than the given threshold but it will never
accept a program which is inconsistent with the input-output example).
This is standard practice in the synthesis
literature~\cite{wang2017program, polozov2015flashmeta}.  

\subsection{Intersection of Two CFTAs}

\noindent Given two CFTAs $\cA_1$ and $\cA_2$ built over the same grammar $G$
from input-output examples $(\sigma_1, o_1)$ and $(\sigma_2, o_2)$
respectively, the intersection of these two automata $\cA$ contains
programs which satisfy both input-output examples (or has the empty language).
Given CFTAs $\cA = (Q, F, Q_f, \Delta)$ and 
$\cA' = (Q', F', Q'_f, \Delta)$, $\cA^* = (Q^*, F, Q^*_f, \Delta^*)$ is the
intersection of $\cA$ and $\cA'$, where $Q^*, Q^*_f,$ and $\Delta^*$ are the
smallest set such that:
\[
    q^{\vec{c_1}}_s \in Q \text{ and } q^{\vec{c_2}}_s \in Q' \text{ then }
    q^{\vec{c_1}:\vec{c_2}}_s \in Q^*
\]
\[
q^{\vec{c_1}}_s \in Q_f \text{ and } q^{\vec{c_2}}_s \in Q'_f \text{ then }
    q^{\vec{c_1}:\vec{c_2}}_s \in Q^*_f
\]
\[
    f(q^{\vec{c_1}}_{s_1}, \ldots q^{\vec{c_k}}_{s_k}) \rightarrow q^{\vec{c}}_s 
    \in \Delta \text{ and } 
 f(q^{\vec{c'_1}}_{s_1}, \ldots q^{\vec{c'_k}}_{s_k}) \rightarrow
 q^{\vec{c'}}_s \in \Delta' 
\]
\[
    \text{then }  f(q^{\vec{c_1}:\vec{c'_1}}_{s_1}, \ldots
    q^{\vec{c_k}:\vec{c'_k}}_{s_k})  \rightarrow q^{\vec{c}:\vec{c'}}_s\in
    \Delta^* 
\]
where $\vec{c}$ denotes a vector of values and $\vec{c_1}:\vec{c_2}$ 
denote a vector constructed  by appending vector
$\vec{v_2}$ at the end of vector $\vec{v_1}$.

\section{Loss Functions}\label{sec:lossFunction}
Given a data set 
$\cD = \{ (\sigma_1, o_1 ), \ldots, ( \sigma_n, o_n ) \}$
and a program $p$, a {\bf Loss Function} $\cL(p,\cD)$ 
calculates how incorrect the program is with respect to the given
data set. We work with loss functions $\cL(p, \cD)$ that depend only on the data set and the
outputs of the program for the inputs in the data set, i.e., given programs $p_1, p_2$, such
that for all $(\sigma_i, o_i) \in \cD$, $\llbracket p_1 \rrbracket \sigma_i =
\llbracket p_2 \rrbracket \sigma_i$, then $\cL(p_1, \cD) = \cL(p_2, \cD)$. 
We also further assume that the loss function $\cL(p, \cD)$ can be expressed in 
the following form:
\[
    \cL(p, \cD) = \sum\limits_{i=1}^n L(o_i, \llbracket p \rrbracket \sigma_i)
\]
where $L(o_i, \llbracket p \rrbracket \sigma_i)$ is a per-example loss function.

\Comment{
XXX what is this XXX
\[
    L(p, \cD) = \cL({\tup{\llbracket p \rrbracket \sigma_i, 
    o_i}\vert (\sigma_i, o_i) \in \cD})
\]

XXX We often work with loss functions  $\cL(p, \cD)$ of the form XXX

\Comment{
Given an input-output example $(\sigma_i, o_i)$, we say the function $l_i$ is 
the output weight function
if 
\[
    l_i(o) = l(o_i, o)
\]
}
}

\begin{definition}
{\bf $0/1$ Loss Function:} The $0/1$ loss function \\ 
$\cL_{0/1}(p, \cD)$ counts the number of
    input-output examples where $p$ does not agree with the data set $\cD$:
    \[
        \cL_{0/1}(p, \cD) = \sum\limits_{i=1}^n 1 ~~\mathrm{if}~~ (o_i \neq \llbracket p \rrbracket \sigma_i) ~~\mathrm{else}~~ 0
    \]
\end{definition}

\begin{definition}
{\bf $0/\infty$ Loss Function:} The $0/\infty$ loss function \\ 
$\cL_{0/\infty}(p, \cD)$ is 0 if $p$ matches all outputs in the data set $\cD$ and
$\infty$ otherwise:

    \[
        \cL_{0/\infty}(p, \cD) = 0 ~~\mathrm{if}~~ (\forall (\sigma, o) \in \cD . o = \llbracket p \rrbracket \sigma) ~~\mathrm{else}~~~ \infty
    \]
\end{definition}

\begin{definition}
{\bf Damerau-Levenshtein (DL) Loss Function:} The DL loss function 
$\cL_{DL}(p, \cD)$ uses the {\it Damerau-Levenshtein} metric~\cite{damerau1964technique},
to measure the distance between the output from the synthesized program and the
corresponding output in the noisy data set: 
\[
    \cL_{DL}(p, \cD) = \sum\limits_{(\sigma_i, o_i) \in \cD} 
    L_{\llbracket p \rrbracket \sigma_i, o_i}\big(\left\vert \llbracket p \rrbracket
    \sigma_i \right\vert, \left\vert o_i \right\vert\big)
\]
where, $L_{a, b}(i, j)$ is the {\it Damerau-Levenshtein} metric~\cite{damerau1964technique}.
\end{definition}
    This metric counts the number of 
single character deletions, insertions, substitutions, or transpositions
required to convert one text string into another. 
Because more than $80\%$ of all human misspellings 
are reported to be captured by a single one of these four 
operations~\cite{damerau1964technique}, the DL loss function may
be appropriate for computations that work with human-provided
text input-output examples. 

\Comment{
\begin{figure*}
    \centering
\begin{equation*}
    d_{a, b}(i, j) = 
    \min 
    \begin{cases}
    0 & i = j = 0 \\
        d_{a, b}(i - 1, j) + 1 & i > 0 \text{ \#Deletion} \\
        d_{a, b}(i, j - 1) + 1 & j > 0 \text{ \#Insertion}\\
        d_{a, b}(i - 1, j - 1) + \mathds{1}(a_i \neq b_k) & i, j > 0 \text{
            \#Substitution}\\
        d_{a, b}(i - 2, j - 2) + 1 & i,j > 1, a_i = b_{j-1}, a_{i-1} = b_j 
        \text{ \#Transposition}
        \end{cases}
\end{equation*}
\centering
    \caption{Damerau-Levenshtein Distance}
    \label{fig:DLdistance}
\end{figure*}
}

\Comment{
\begin{example}{\bf ($0/1$ loss function)}
    Consider the following $0/1$ loss function, which counts the number of
    input-output examples where $p$ does not agree with the data set $\cD$.
    \[
        L_{0/1}(p, \cD) = \sum\limits_{i=1}^n 1 ~~\mathrm{if}~~ (o_i \neq \llbracket p \rrbracket \sigma_i) ~~\mathrm{else}~~ 0
    \]
    where the (per-example) loss function $\mathds{1}(o \neq \llbracket p \rrbracket
    \sigma)$ is the indicator function.
\end{example}
}

\section{Complexity Measure}
\label{sec:complexityMeasure}
Given a program $p$, a {\bf Complexity Measure} $C(p)$ ranks 
programs independent of the input-output examples in the 
data set $\cD$.  This measure is used to trade off performance on the noisy data set vs. complexity
of the synthesized program.  Formally, a complexity measure is a function
$C(p)$ that maps each program $p$ expressible in the given
DSL $G$ to a real number. The following $\mathrm{Cost}(p)$ complexity
measure computes the complexity of given program $p$ represented as a parse tree recursively as follows:
\[
\begin{array}{rcl}
\mathrm{Cost}(t) &=& \mathrm{cost}(t)\\
\mathrm{Cost}(f(e_1, e_2, \ldots e_k))  &=& \mathrm{cost}(f) + \sum\limits_{i = 1}^k \mathrm{Cost}(e_i) 
\end{array}
\]
where $t$ and $f$ are terminals and  built-in functions in our DSL respectively.
Setting $\mathrm{cost}(t) = \mathrm{cost}(f) = 1$ delivers a complexity measure 
$\mathrm{Size}(p)$ that computes the size of $p$.  

Given an FTA $\cA$, we can use dynamic programming to find the minimum complexity parse tree 
(under the above $\mathrm{Cost}(p)$ measure) accepted by $\cA$~\cite{gallo1993directed}. 
In general, given an FTA $\cA$, we assume 
we are provided with a method to find the program $p$ accepted by $\cA$ which
minimizes the complexity measure. 

\Comment{

XXX - should be size complexity for consistency with results section - XXX
\begin{example}{\bf (Size complexity measure)}
    Consider the complexity function $Cost_s$:
 \[
\begin{array}{rcl}
    Cost_s(t) &=& 1\\
Cost_s(f(e_1, e_2, \ldots e_k))  &=& 1 + \sum\limits_{i = 1}^k Cost_s(e_i) 
\end{array}
\]
where $t$ and $f$ are terminals and  built-in functions in our DSL respectively.
$Cost_s$ measures the size of input parse tree.
\end{example}
}

\section{Objective Functions}
Given loss $l$ and complexity $c$, an {\bf Objective Function} $U(l,c)$ maps $l,c$ 
to a totally ordered set such that for all $l$,
$U(l, c)$ is monotonically nondecreasing with respect to $c$. 

\begin{definition}
{\bf Tradeoff Objective Function:}
    Given a tradeoff parameter $\lambda > 0$, 
    the tradeoff objective function $U_{\lambda}(l, c) = l + \lambda c$.
\end{definition}

This objective function trades the loss of the synthesized program off against
the complexity of the synthesized program. 
Similarly to how regularization can prevent a machine learning model from overfitting noisy data by biasing the training algorithm to
pick a simpler model, the tradeoff objective function may prevent the algorithm from
synthesizing a program which overfits the data by biasing it to pick a simpler
program (based on the complexity measure).

\begin{definition}
{\bf Lexicographic Objective Function:}
   A lexicographic objective function
   $U_L(l,c) = \tup{l,c}$ maps $l$ and $c$ 
   into a lexicographically ordered space, i.e., $\tup{l_1, c_1} < \tup{l_1, c_2}$ if and only
   if either $l_1 < l_2$ or $l_1 = l_2$ and $c_1 < c_2$.
\end{definition}
This objective function first minimizes the loss, then the complexity. It may be appropriate, for
example, for best fit program synthesis, data cleaning and correction, and approximate program synthesis
over clean data sets. 

\section{State-Weighted FTA}
State-weighted finite tree automata (SFTA) are FTA augmented with a weight function 
that attaches a weight to all accepting states. 

\begin{definition}[\bf SFTA]
    A state-weighted finite tree automaton  (SFTA) over alphabet $F$ is a tuple
    $\cA = (Q, F, Q_f, \Delta, w)$ where $Q$ is a set of states, $Q_f
    \subseteq Q$ is the set of accepting states, $\Delta$ is a set of
    transitions of the form $f(q_1, \ldots, q_k) \rightarrow q$ where 
    $q, q_1, \ldots q_k$ are states, $f \in F$ and $w : Q_f \rightarrow
    \mathds{R}$ is 
    a function which assigns a weight $w(q)$ (from domain $W$) to each accepting 
    state $q \in Q_f$.
\end{definition}

Because CFTAs are designed to handle synthesis over clean (noise-free) data sets, 
they have only one accept state $q^o_{s_0}$ (the state with start symbol $s_0$
and output value $o$). We weaken this condition to allow multiple accept states 
with attached weights using SFTAs.

\begin{figure}
    \[
        \begin{array}{c}
        \begin{array}{cc}
            \infral{t \in T, ~~~~~ \llbracket t \rrbracket \sigma = c}
            {q^c_t \in Q}
            {(Term)}
            &
            \infral{q^{c}_{s_0} \in Q}
            {q^{c}_{s_0} \in Q_f \\ w(q^c_{s_0}) = L(o, c)}
            {(Final)}
            \\
         \end{array}
            \\
            \\
            \infral{s \rightarrow f(s_1, \ldots s_k) \in P, ~~~~~~~~~ 
            \{ q_{s_1}^{c_1}, \ldots, q_{s_k}^{c_k} \} \subseteq Q,\\
            \llbracket f(c_1, \ldots c_k) \rrbracket\sigma = c}
            {q_s^c \in Q, ~~~~~ f(q_{s_1}^{c_1}, \ldots, q_{s_k}^{c_k}) 
            \rightarrow q_s^c \in \Delta}{(Prod)}
        \end{array}
    \]
    \caption{Rules for constructing a SFTA $\cA = (Q, F, Q_f, \Delta, w)$
    given input $\sigma$, per-example loss function $L$, and grammar $G = (T, N, P, s_0)$.}
    \label{fig:wfta_rules}
\end{figure}

Given an input-output example $(\sigma, o)$ and per-example loss function $L(o,c)$, 
Figure~\ref{fig:wfta_rules} presents rules for constructing a SFTA
that, given a program $p$, returns the loss for $p$ on example $(\sigma, o)$.
The SFTA $\mathsf{Final}$ rule (Figure~\ref{fig:wfta_rules}) marks all states 
$q^c_{s_0}$ with start symbol $s_0$ as accepting states regardless of the
concrete value $c$ attached to the state. The rule also associates the loss $L(o,c)$ for concrete
value $c$ and output $o$ with state $q^c_{s_0}$ as the weight $w(q^c_{s_0}) = L(o,c)$. 
The CFTA $\mathsf{Final}$ rule (Figure~\ref{fig:cfta_rules}), in contrast, marks only 
the state $q^o_{s_0}$ (with output value $o$ and start state $s_0$) as the accepting state.

A SFTA divides the set of programs in the DSL into subsets. 
Given an input $\sigma$, all programs in a subset produce the same
output (based on the accepting state), with the SFTA assigning a
weight $w(q^c_{s_0}) = L(o,c)$ as the weight of this subset of programs. 

We denote the SFTA constructed from DSL $G$, example 
$(\sigma, o)$, per-example loss function $L$, and threshold $d$  as
$\cA_G^d(\sigma, o, L)$. We omit the subscript grammar $G$ and threshold $d$
wherever it is obvious from context.

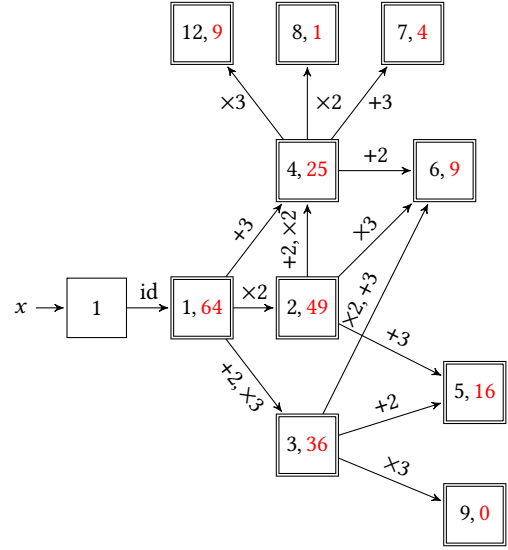
\begin{figure}[tb]
        \begin{tikzpicture}[->,>=stealth',shorten >=1pt,auto,node distance=1.4cm]

                        \tikzset{every state/.append style={rectangle}}
 \node[state, initial by arrow, initial text={$x$}]  (x) {$1$};
            \node[state, accepting]           (1)  [right of=x] {$1,
            {\color{red} 64}$};
            \node[state, accepting]          (2)   [right of=1]    {$2,
            {\color{red} 49}$};
            \node[state, accepting]    (4) [above=1.0cm of 2]    {$4,
            {\color{red} 25}$};
            \node[state, accepting]    (3) [below=1.0cm of 2]    {$3,
            {\color{red} 36}$};
            \node[state, accepting]   (6) [right=1.0cm of 4] {$6, {\color{red} 9}$};
            \node[state, accepting]   (5) [below right=0.3cm and 1.4cm of 2]
            {$5, {\color{red} 16}$};
            \node[state, accepting]   (9) [below right=0.1cm and 1.4cm of 3]
            {$9, {\color{red} 0}$};
            \node[state, accepting]   (8)  [above=1.0cm of 4]  {$8, {\color{red}
            1}$};
            \node[state, accepting]   (12) [left of=8]  {$12, {\color{red} 9}$};
            \node[state, accepting]   (7)  [right of=8] {$7, {\color{red} 4}$};

            \path (x)   edge node [] {id} (1)
                    (1)   edge node [sloped, below]{$+2, \times 3$} (3)
                        edge node [sloped, above] {$\times 2$}  (2)
                        edge node [sloped, above] {$+3$}  (4)
                  (2)   edge node [sloped, above] {$+2, \times 2$} (4)
                        edge node [sloped, above] {$\times 3$} (6)
                        edge node [sloped, above] {$+3$} (5)
                  (3)   edge node [sloped, above]{$\times 2, + 3$} (6)
                        edge node [sloped, above]{$+2$} (5)
                        edge node [sloped, above] {$\times 3$} (9)
                  (4)   edge node [right] {$\times 2$} (8)
                        edge node [sloped, above] {$+2$} (6)
                        edge node [left] {$\times 3$} (12)
                        edge node [right] {$+3$} (7)
            ;
        \end{tikzpicture}
    \caption{The SFTA constructed for Example~\ref{ex:wfta1}}
    \label{fig:wfta1}
\end{figure}

\begin{example}\label{ex:wfta1}
    Consider the DSL presented in Example~\ref{ex:dsl}.
    Given input-output example $(\{x \rightarrow 1\}, 9)$ and weight function
    $l(c) =  (c - 9)^2$, Figure~\ref{fig:wfta1} presents the
    SFTA which represents all
    programs of height less than $3$.
\end{example}
For readability, we omit the states for terminals $2$ and $3$. For all
accepting states the first number (the number in black) represents the computed
value and the second number (the number in red) represents the weight of the
accepting state.

\subsection{Operations over SFTAs}

\begin{definition}[\bf $+$ Intersection]
Given two SFTAs \\$\cA_1 = (Q_1, F, Q^1_f, \Delta_1, w_1)$ and
    $\cA_2 = (Q_2, F, Q^2_f, \Delta_2, w_2)$, a SFTA $\cA = 
    (Q, F, Q_f, \Delta, w)$ is the $+$ intersection $\cA_1$ 
    and $\cA_2$, if the CFTA in $\cA$ is the intersection of
    CFTAs of $\cA_1$ and $\cA_2$, and the weight of 
    accept states in $\cA$ is the sum of weight of corresponding 
    weights in $\cA_1$ and $\cA_2$. Formally:
    \begin{itemize}
    \item The CFTA
    $(Q, F, Q_f, \Delta)$ is the intersection 
    of CFTAs \\ $(Q_1, F, Q^1_f, \Delta_1)$ and $(Q_2, F, Q^2_f, \Delta_2)$
    \item
        $w(q^{\vec{c_1}:\vec{c_2}}_s) = w_1(q^{\vec{c_1}}_s) +
            w_2(q^{\vec{c_2}}_s)$ (for $q^{\vec{c_1}:\vec{c_2}}_s \in Q_f$).
    \end{itemize}
\end{definition}

\noindent{Given} two SFTAs $\cA_1$ and $\cA_2$, $\cA_1 + \cA_2$
denotes the $+$ intersection of $\cA_1$ and $\cA_2$.

\begin{definition}[\bf $/$ Intersection]
    Given a SFTA \\ $\cA = (Q, F, Q_f, \Delta, w)$ and a
    CFTA $\cA^* = (Q^*, F, Q^*_f, \Delta^*)$, 
    a SFTA $\cA' = (Q', F, Q'_f, \Delta', w')$ is the 
    $/$ intersection of $\cA$ and $\cA^*$, if the FTA $\cA'$ is the 
    intersection of CFTA $\cA$ and $\cA^*$, and the weight of the 
    accepting state in $\cA'$ is the same as the weight of the
    corresponding accepting state in $\cA$. Formally:

    \begin{itemize}
    \item The CFTA
    $(Q', F, Q'_f, \Delta')$ is the intersection 
    of FTAs \\ $(Q, F, Q_f, \Delta)$ and $(Q^*, F, Q^*_f, \Delta^*)$
    \item
        $w'(q^{\vec{c_1}:\vec{c_2}}_s) = w(q^{\vec{c_1}:\vec{c_2}}_s)$ (for
            $q^{\vec{c_1}:\vec{c_2}}_s \in Q'_f$).
    \end{itemize}
\end{definition}

\noindent 
Given a SFTA $\cA$ and a CFTA $\cA^*$, $\cA / \cA^*$
denotes the $/$ intersection of $\cA$ and $\cA^*$.

Given a single input-output example, a CFTA built on that example only 
accepts programs which are consistent with that example. $/$
intersection is a simple method to prune a SFTA to only contain programs which
are consistent with an input-output example.

\begin{definition}[\bf $w_0$-pruned SFTA]
A SFTA \\ $\cA' =  (Q, F, Q'_f, \Delta, w')$ is the $w_0$-pruned 
SFTA of \\$\cA = (Q, F, Q_f, \Delta, w)$ if we remove all accept states with
weights greater $w_0$ from $Q_f$. Formally, $Q'_f = \{q \vert q \in Q_f \wedge
w_0 <= w(q)  \}$ and $w'(q) = w(q)$ if $q \in Q'_f$.
\end{definition}

\noindent
Given a SFTA $\cA$, $\cA \downarrow_{w_0}$ denotes the $w_0$-pruned SFTA of $\cA$.

\begin{definition}[\bf $q$-selection]
    Given a SFTA \\ $\cA = (Q, F, Q_f, \Delta, w)$ and a accept state $q \in Q_f$, 
    the CFTA $(Q, F, \{q\}, \Delta)$ is called the $q$-selection of SFTA $\cA$.
\end{definition}
\noindent
Given a SFTA $\cA$, the notation $\cA_q$ denotes the $q$-selection of SFTA $\cA$.

\section{Synthesis Over Noisy Data}
\label{sec:syn_noise}
Given a data set 
$ \{ (\sigma_1, o_1 ), \ldots, ( \sigma_n, o_n ) \}$ of input-output examples
and loss function $\cL(p, \cD)$ with per-example loss function $L$, 
we construct SFTAs for each input-output example $\cA_1, \cA_2,
\ldots \cA_n$ where $\cA_i = \cA(\sigma_i, o_i, L)$.

\begin{theorem}\label{thm:wfta_l}
    Given a SFTA $\cA = \cA(\sigma, o, L) = (Q, F, Q_f, \Delta, w)$, for all
    accepting states $q \in Q_f$ and for all programs $p$ accepted by the
    $q$-selection automata $\cA_q$:
    \[
        L(o, \llbracket p \rrbracket \sigma) = w(q)
    \]
\end{theorem}
\begin{proof}
    Consider any state $q \in Q_f$.
    All programs accepted by state $q$ compute the same concrete value $c$
    on the given input $\sigma$.
    Hence for all programs accepted by the $q$-selection automata $\cA_q$, 
    $
        \llbracket p \rrbracket \sigma = c
    $.
    By construction (Figure~\ref{fig:wfta_rules}), 
    $
        w(q) = L(c) = L(\llbracket p \rrbracket \sigma) 
    $
\end{proof}

Let SFTA $\cA(\cD, L)$ be the $+$ intersection of 
SFTAs defined on input-output examples in $\cD$. Formally:
\[
\cA(\cD, L) = \cA(\sigma_1, o_1, L) + \cA(\sigma_2, o_2, L) + \ldots \cA(\sigma_n,
o_n, L)
\]
Since the size of each SFTA $\cA(\sigma_i, o_i, L)$ is bounded, the cost of computing $\cA(\cD, L)$ is
$\mathcal{O}(|\cD|)$.

\begin{theorem}\label{thm:wfta_L}
Given $\cA(\cD, L) = (Q, F, Q_f, \Delta, w)$ as defined above, for all accepting
states $q \in Q_f$, for all programs $p$ accepted by the $q$-selection automata
$\cA(\cD, L)_q$:
\[
    \cL(p, \cD) = w(q)
\] 
i.e., the weight of the state $q$ measures the loss of programs by $q$ on
    data set $\cD$.
\end{theorem}
\begin{proof}
    Consider any accepting state $q \in Q_f$. 
    Since $\cA(\cD, L)$ is an intersection of SFTAs $\cA_1$ \ldots 
    $\cA_n$ (where $\cA_i = \cA(\sigma_i, o_i, L) = 
    Q_i, F, (Q_f)_i, \Delta_i, w_i)
    $), there exist accepting states\\ $q_1 \in (Q_f)_1, q_2 \in (Q_f)_2, \ldots 
    q_n \in (Q_f)_n$ such that all programs $p$ accepted by $\cA_q$ are 
    accepted by $(\cA_1)_{q_1}, (\cA_2)_{q_2} \ldots (\cA_n)_{q_n}$.
    
    \noindent From Theorem~\ref{thm:wfta_l}, for all programs $p$ accepted by $\cA_q$,
    $ 
        w_i(q_i) = L(o_i, \llbracket p \rrbracket \sigma_i) 
    $
    From definition of $+$ intersection,
    \[
        w(q) = \sum\limits_{i=1}^n w_i(q_i) = \sum\limits_{i=1}^n 
        L(o_i, \llbracket p \rrbracket \sigma_i) = \cL(p, \cD)
    \]
\end{proof}

 \begin{algorithm}
     \SetAlgoLined
     \SetKwInOut{Input}{Input}
     \Input{DSL $G$, threshold $d$, data set $\cD$, per-example loss function $L$, 
     complexity measure $C$, and objective function $U$}
\KwResult{Synthesized program $p^*$}
     $\cA(\cD, L) = (Q, F, Q_f, \Delta, w)$\\ 
     \#the SFTA over data set $\cD$ and per-example loss function $L$\\
     \ForEach{$q \in Q_f$}{
         $p_q \leftarrow \mathsf{argmin}_{p \in \cA(\cD, L)_q} C(p)$\\
         \# For each accepting state $q$, find the most optimal program $p_q$
     }

     $q^* \leftarrow \mathsf{argmin}_{q \in Q_f} 
     U(w(q), C(p_q))$\\
    $p^* \leftarrow p_{q^*}$
 \caption{Synthesis Algorithm}
     \label{alg:regsyn}
 \end{algorithm}

Algorithm~\ref{alg:regsyn} presents the base algorithm to
synthesize programs within various noisy synthesis settings.

\begin{theorem}\label{thm:main}
    The program $p^*$ returned by Algorithm~\ref{alg:regsyn}
    is equal to $p'$ where 
    \[
        p' = \mathsf{argmin}_{p \in G_d}
        U(\cL(p, \cD), C(p))
    \] 
\end{theorem}
\begin{proof}
    Given $\cA(\cD, L) = (Q, F, Q_f, \Delta, w)$. 
    $p^*$ returned by Algorithm~\ref{alg:regsyn} is equal to $p_{q^*}$, where
    $q^* =
    \mathsf{argmin}_{q \in Q_f} U(w(q), C(p_q))$, where $p_q =
    \mathsf{argmin}_{p \in \cA(\cD, L)_q} C(p)$.
    We can rewrite $q^*$ as
    \[
        \mathsf{argmin}_{q \in Q_f} U(w(q), \min_{p \in \cA(\cD, L)_q}C(p))
    \]
    Since for any $l$, $U(l, c)$ is non-decreasing with respect to $c$, 
    we can rewrite $q^*$ as 
    \[
        \mathsf{argmin}_{q \in Q_f} \min_{p \in
        \cA(\cD, L)_q} U(w(q), C(p))
    \]
    By Theorem~\ref{thm:wfta_L}, for any $p \in \cA(\cD, L)_q$:
    \[
        w(q) = \cL(p, \cD)
    \]
    \[
        q^* = \mathsf{argmin}_{q \in Q_f} \min_{p \in
        \cA(\cD, L)_q} U(\cL(p, \cD), C(p))
    \]
    Because $q^*$ is the accepting state of $p^*$ and $p \in \cA(\cD, L)$ if and
    only if $\exists~q\in Q_f. p \in \cA(\cD, L)_q$, we can rewrite the above equation as:
    \[
        p^* = \mathsf{argmin}_{p \in \cA(\cD, L)}  U(\cL(p, \cD), C(p)) 
    \]
    The set of programs accepted by $\cA(\cD, L)$ is the same set of programs in
    grammar $G_d$.
    Hence proved.
    \end{proof}

We next present several modifications of the core algorithm to solve various
synthesis problems.

\subsection{Accuracy vs. Complexity Tradeoff}
\label{subsec:acc_comp}
Given a DSL $G$, a data set $\cD$,  loss function $\cL$,
complexity measure $c$, and positive weight $\lambda$, we wish 
to find a program $p^*$ which minimizes the weighed sum of the loss function and
the complexity measure. Formally:
\[
    p^* = \mathsf{argmin}_{p \in G_d}(\cL(p, \cD) + \lambda \cdot C(p))
\]
where $G_d$ is the set of programs in DSL $G$ with size less than the threshold
$d$. By using the objective function $U(l, c) = l + \lambda c$, we can use
Algorithm~\ref{alg:regsyn} to synthesize program $p^*$ which minimizes the
objective function given above.

\subsection{Best Program for Given Accuracy}
\label{subsec:best}
Given a DSL $G$, a data set $\cD$, loss function $\cL$, complexity
measure $C$ and bound $b$, we wish to find a program $p^*$ that minimizes 
the complexity measure $C$ but has loss less than $b$.
Formally:
$p^* = \mathsf{argmin}_{p \in G_d} C(p)~\text{ s.t. } \cL(p, \cD) < b$. 
Note that this condition can be rewritten as 
\[
    p^* = \mathsf{argmin}_{p \in \cA'} C(p)
\]
where $\cA' = \cA(\cD, L) \downarrow_b$.

By the definition of $\downarrow_b$, all accepting states of $\cA'$
have weight less than $b$. Therefore all programs accepted by $\cA'$ have loss
less than $b$ (i.e. $\cL(p, \cD) < b$). Also note that if a program $p$ is not in
$\cA'$ then either it has loss greater than $b$ or it is not within the threshold $d$. 

\subsection{Forced Accuracy}
Given DSL $G$, a data set $\cD$, a subset $\cD' \subseteq \cD$, 
loss function $\cL$,
complexity measure $C$, and objective function $U$, we wish 
to find a program $p^*$ which minimizes the objective function with an added
constraint of bounded loss over data set $\cD'$. Formally:
\[
    p^* = \mathsf{argmin}_{p \in G_d} U(\cL(p, \cD), C(p)) \text{ s.t. } 
    \cL(p, \cD') \leq b
\]

We first construct a SFTA $\cA(\cD', L) \downarrow_b$ which contains all
programs consistent with loss less than or equal to $b$ over data set $\cD'$. 
After constructing $\cA(\cD, L)$ as in Algorithm \ref{alg:regsyn}, we
modify $\cA(\cD, L)$ by $/$ intersection $\cA(\cD', L)$ (after dropping the weights of
the accepting states) with $\cA(\cD, L)$ (i.e. $\cA(\cD, L)
\leftarrow \cA(\cD, L) / \cA(\cD', L)$ as in Algorithm~\ref{alg:regsyn}).
By definition of $/$ intersection and $\cA$, loss of all programs returned 
by the modified algorithm on $\cD'$ will be less than equal to $b$.

\vspace{-.1in}
\section{Use Cases}

\begin{definition}
{\bf Bayesian Program Synthesis}: 
Given a set of input-output examples $\cD = \{
(\sigma_i, o_i) \vert i = 1 \ldots n\}$, DSL grammar $G$, and a probability
distribution $\pi$, $p^*$ is the solution to the Bayesian
program synthesis problem, if $p^*$ is the most probable program in DSL $G$,
given the data set $\cD$. Formally
$p^* = \mathsf{argmax}_{p \in G} \pi(p~\vert~\cD)$.
\end{definition}

\Comment{
To restrict the search domain, we restrict the above problem by adding an
additional constraint that the size of the program $p^*$ is less than a
threshold $d$, i.e.
\[
p^* = \mathsf{argmax}_{p \in G} \pi(p~\vert~\cD~\text{ and size of } p \leq d)  
\]
}
\noindent By Bayes rule $p^*  = \mathsf{argmax}_{p \in G} \pi(\cD \vert p)\pi(p)$, so 
\[
p^* = \mathsf{argmax}_{p \in G} \big[ 
  (\log \pi(\cD \vert p))
  + (\log \pi(p))\big]\] 
Assuming independence of observations: 
\Comment{
where $G_d$ represents the set of programs $p$ which can be specified by 
grammar $G$ with size less than $d$. 
\[
p^* = \mathsf{argmax}_{p \in G_d} \big[ 
  (\log \pi(\cD \vert p))
  + (\log \pi(p))\big]
\]
}
\[
p^* = \mathsf{argmax}_{p \in G_d} \big[ 
\Big( \sum_{(\sigma_i, o_i) \in \cD} 
 \log \pi(o_i \vert \llbracket p \rrbracket \sigma_i)
\Big)
+ (\log \pi(p))
\big]
\]
Where $\pi(o_i \vert \llbracket p \rrbracket \sigma_i)$ denotes the probability
of output observation $o_i$ in the data set $\cD$, given a program $p$ 
With complexity measure  $\log \pi(p)$, per-example loss function $\log \pi(o_i
\vert \llbracket p \rrbracket \sigma_i)$ (given example $(\sigma_i, o_i)$), 
and the following loss function: 
\[
\cL(p, \cD) = \sum_{(\sigma_i, o_i) \in \cD}  
\log \pi(o_i \vert \llbracket p \rrbracket \sigma_i))
\]
the technique in Section~\ref{subsec:acc_comp}
(Algorithm~\ref{alg:regsyn}) synthesizes the most probable program $p^*$

\noindent{\bf At Most $k$ Wrong:}
Consider a setting in which, given a data set, a random procedure is allowed to
corrupt at most $k$ of these input-output examples. Given this noisy data set
$\cD$, our task is to synthesize the simplest program $p^*$ which is wrong on
at most $k$ of these input-output examples. Formally, given data set $\cD$, bound
$k$, DSL $G$, and a complexity measure $C$:
\[
p^* = \mathsf{argmin}_{p \in G} C(p) \text{ s.t. } \cL_{0/1}(p, \cD) \leq k
\]  
where $\cL_{0/1}$ is the 0/1 loss function. The best program for a given
accuracy framework (subsection~\ref{subsec:best}) allows us to synthesize $p^*$ subject
to a threshold $d$.

\vspace{-.1in}
\section{Experimental Results}
\label{sec:results}
String transformations have been extensively studied within the 
Programming by Example community~\cite{gulwani2011automating, polozov2015flashmeta, singh2016transforming}. 
We implemented our technique (in 6k lines of Java code) and used it to solve benchmark program synthesis
problems from the SyGuS 2018 benchmark suite~\cite{alur2013syntax}.
This benchmark suite contains a range of string transformation problems,
a class of problems that has been extensively studied in past
program synthesis 
projects~\cite{gulwani2011automating, polozov2015flashmeta, singh2016transforming}.

We use the DSL from~\cite{wang2017program} (Figure~\ref{fig:str_lang})
with the size complexity measure $\mathrm{Size}(p)$. 
The DSL supports extracting and contatenating 
($\mathsf{Concat}$) substrings of the input string $x$; each substring
is extracted using the $\mathsf{SubStr}$ function with a start and end
position. A position can either be a constant index ($\mathsf{ConstPos}$) or 
the start or end of the $k^{th}$ occurrence of the match token $\tau$ 
in the input string ($\mathsf{Pos}$).

\begin{figure}
    \[
        \begin{array}{rcl}
            \text{String expr} ~~~~ e &:=&
            \mathsf{Str}(f)~\vert~\mathsf{Concat}(f, e);\\
            \text{Substring expr} ~~~~ f &:=& \mathsf{ConstStr}(s)~\vert~
            \mathsf{SubStr}(x, p_1, p_2);\\
            \text{Position} ~~~~ p &:=& \mathsf{Pos}(x, \tau, k, d)~\vert~
            \mathsf{ConstPos}(k)\\
            \text{Direction} ~~~~ d &:=& \mathsf{Start}~\vert~\mathsf{End};\\
        \end{array}
    \]
    \vspace{-.2in}
    \caption{DSL for string transformation, $\tau$ is a token, $k$ is an
    integer, and $s$ is a string constant}
    \vspace{-.2in}
    \label{fig:str_lang}
\end{figure}

\vspace{-.1in}
\subsection{Implementation Optimizations}

Instead of computing individual SFTAs for each input-output example, then
combining the SFTAs via + intersections to obtain the final SFTA, our implementation
computes the final SFTA directly working over the full data set. 
The implementation also applies two techniques that constrain the size
of the final SFTA. 
First, it bounds the number of recursive applications of the production $e :=
\mathsf{Concat}(f, e)$ by applying a {\it bounded scope height threshold} $d$.
Second, during construction of the SFTA, a state with symbol $e$ is only added 
to the SFTA if the length of the state's output value is not greater than the length of
the output string plus one.

\Comment{
In general constructing individual SFTAs and then intersecting them may
introduce opportunities to reduce the number of states at any given point in the
combined SFTA (specially in cases where some of the individual SFTAs were $w_0-$pruned).
To optimize our technique for the lexicographic objective function and tradeoff
objective function in absence of any pruning, 
our implementation does not construct the final SFTA by
constructing a combined CFTA on the entire dataset (Section 2.2,
\cite{wang2017program}), adding all states with the start symbol to the accept
states, and then computing the loss on the entire dataset by using the vector of
values attached to each state.
}

\vspace{-.1in}
\subsection{Scalability}
\label{sec:scalability}

\begin{figure*}
    {\small
    \begin{center}
        \begin{tabular}{|l|c|c| c|c| c|c| c|c|}
            \hline
            Threshold & \multicolumn{2}{c|}{1} &  \multicolumn{2}{c|}{2} & 
            \multicolumn{2}{c|}{3} & \multicolumn{2}{c|}{4}   \\ 
            \hline
            Benchmark Name & time(sec) & SFTA size & time(sec) & SFTA size & 
            time(sec) & SFTA size & time(sec) & SFTA size\\
            \hline
            bikes &    0.16&1.08&    0.73&10.56&    4.72&  56.4&  19.83& 145.8\\
            bikes-long  &     0.21&1.02&    1.37&9.42& 6.04&49.9& 26.99&
            139.35\\
            bikes-long-repeat &   0.18& 1.02&  1.06&9.42&   6.03& 49.9& 27.47& 139.35 \\
            bikes-short  &   0.15&1.08&  0.79&10.56&  3.98& 56.4&  18.62& 145.8\\
            dr-name &  X&X&  7.54&107.5&   107.18& 1547.2&  -&-\\
            dr-name-long &  X&X&  17.4&70.28&  300.9& 1077.6&  -&-\\
             dr-name-long-repeat &  X&X&  19.15&70.28&  301.3& 1077.6&  -&-\\
            dr-name-short &  X&X&  10.2&107.5&  101.5& 154.8&  -&-\\
            firstname & 0.28&1.02& 1.46&4.34&  4.024& 4.33& 3.97&4.34\\
            firstname-long & 1.72&1.04& 12.03&4.36& 39.08& 4.36& 41.22& 4.36\\
            firstname-long-repeat & 1.64&1.04& 13.96&4.36&  42.4& 4.36& 43.1& 4.36\\
            firstname-short & 0.26&1.02& 1.47&4.37& 3.93& 4.34& 3.9&4.34\\
            initials & X&X& X&X& 8.7&42.3& 30.4&42.34\\
            initials-long & X&X& X&X& 86.44&42.36& 376.56& 42.36 \\
            initials-long-repeat & X&X& X&X& 86.23& 42.36& 386.25& 42.36\\
            initials-short & X&X& X&X& 8.92& 42.34& 31.72& 42.34\\
            lastname & 0.43&2.56& 4.78&28.3& 27.29&208.35& 159.41& 741.44\\
            lastname-long & 1.93&1.37& 15.1&11.34& 112.04&50.81& 485.98& 50.8\\
            lastname-long-repeat  & 1.85&1.37& 18.35&11.34& 113.36& 50.81& 486.35& 50.8\\
            lastname-short & 0.6&2.56& 3.07&28.3& 28.3& 208.35& 160.54& 741.44\\
            name-combine & X&X& 8.49&269.9& 224.074&7485.83& -&-\\
            name-combine-long & X&X& 32.28&161.54& -&-& -&-\\
            name-combine-long-repeat & X&X& 98.46&299& -&-& -&-\\
            name-combine-short & X&X& 6.5&269.9& 207.86& 7485.83& -&-\\
            name-combine-2 & X&X& X&X& 63.490&855.34& -&-\\
            name-combine-2-long & X&X& X&X& 591.6&851.44& -&-\\
            name-combine-2-long-repeat & X&X& X&X& 592.0&851.44& -&-\\
            name-combine-2-short & X&X& X&X& 57.26&855.34& -&-\\
            name-combine-3 & X&X& X&X& 43.082&911.53& 527.29&8104.7\\
            name-combine-3-long & X&X& X&X& 193.42&649.13& -&-\\
            name-combine-3-long-repeat & X&X& X&X& 192.81&649.13& -&-\\
            name-combine-3-short & X&X& X&X& 42.266& 911.53& 526.13&8104.7\\
            \Comment{
            name-combine-4 & X&X& X&X& X&X& -&-\\
            name-combine-4-long & X&X& X&X& -&-& -&-\\
            name-combine-4-long-repeat & X&X& X&X& -&-& -&-\\
           name-combine-4-short & X&X& X&X& X&X& -&-\\
       }
            reverse-name & X&X& 6.9&269.9& 217.19&7495.9& -&-\\
            reverse-name-long & X&X& 29.55&161.53& -&-& -&-\\
            reverse-name-long-repeat & X&X& 27.6&161.53& -&-& -&-\\
            reverse-name-short & X&X& 6.84&269.9& 228.24&7485.8& -&-\\
            phone & 0.12&0.46& 0.47&1.58& 0.87& 1.58& 0.78& 1.58\\
            phone-long & 0.8&0.46& 3.9&1.58& 7.79& 1.58& 32.79&1.58\\
            phone-long-repeat & 0.69&0.46& 3.29&1.58& 7.76& 1.58& 43.24&1.58\\
            phone-short & 0.12&0.46& 0.37&1.58& 0.804& 1.578& 4.97&1.58\\
            phone-1 & 0.15&0.46& 0.44&1.58& 0.84&1.58& 3.017&1.58\\
            phone-1-long & 0.99&0.46& 3.8&1.58& 8.23&1.58& 16.58&1.58\\
            phone-1-long-repeat & 0.90&0.46& 4.1&1.58& 8.42&1.58& 17.5&1.58\\
            phone-1-short & 0.14&0.46& 0.45&1.58& 0.8&1.58&  1.5&1.58\\
            phone-2 & 0.13&0.46& 0.44&1.58& 0.83& 1.58& 3.176&1.58\\
            phone-2-long & 0.64&0.46& 2.84&1.58& 8.36&1.58& 16&1.58\\
            phone-2-long-repeat & 0.85&0.46& 3.8&1.58& 9.83&1.58& 17.55&1.58\\
            phone-2-short & 0.09&0.46& 0.47&1.58& 0.83&1.58& 2.78&1.58\\
    \Comment{
    phone-3 & X&X& X&X& X&X& -&-\\
             phone-3-long & X&X& X&X& -&-& -&-\\
             phone-3-long-repeat & X&X& X&X& -&-& -&-\\
             phone-3-short & X&X& X&X& X&X& -&-\\
             phone-4 & X&X& X&X& X&X& -&-\\
            phone-4-long & X&X& X&X& -&-& -&-\\
             phone-4-long-repeat & X&X& X&X& -&-& -&-\\
             phone-4-short & X&X& X&X& X&X& -&-\\
         }
           phone-5 & 0.18&0.23& 0.16&0.23& 0.11&0.23& 0.7& 0.23\\
            phone-5-long & 1.24&0.23& 0.94&0.23& 0.75&0.23& 4.2& 0.23\\
            phone-5-long-repeat &1.27&0.23& 1.19&0.23& 0.77&0.23& 2.96& 0.23\\
            phone-5-short& 0.17&0.23& 0.17&0.23& 0.11&0.23& 0.9& 0.23\\
            phone-6 & 0.27&0.64& 1.38&2.6& 2.67&2.61& 9.3&2.61\\
            phone-6-long & 1.84&0.64& 6.48&2.6& 24.66&2.61& 103.3&2.61\\
            phone-6-long-repeat & 2.16&0.64& 7.12&2.6& 24.69&2.61& 143.9&2.61\\
            phone-6-short & 0.28&0.64& 0.76&2.6& 2.27&2.61& 11.19&2.61\\
            phone-7 & 0.24&0.64& 1.04&2.6& 2.87&2.61& 11.141&2.61\\
            phone-7-long & 2.6&0.64& 7.8&2.6& 26.1&2.61& 108.1&2.61\\
            phone-7-long-repeat & 2.58&0.64& 6.68&2.6& 26.15&2.61& 115.42&2.61\\
            phone-7-short & 0.23&0.64& 1.13&2.6& 3.26&2.61& 10.71&2.61\\
            phone-8 & 0.23&0.64& 1&2.6& 2.65&2.61& 8.51&2.61\\
            phone-8-long & 2.33&0.64& 7.58&2.6& 25.87&2.61& 114.54&2.61\\
            phone-8-long-repeat & 1.67&0.64& 7.7&2.6& 25.45&2.61& 128.3&2.61\\
            phone-8-short & 0.27&0.64& 0.97&2.6& 2.45&2.61& 13.81&2.61\\
           \Comment{
           phone-9 & X&X& X&X& -&-& -&-\\
            phone-9-long & X&X& X&X& -&-& -&-\\
            phone-9-long-repeat& X&X& X&X& -&-& -&-\\
            phone-9-short & X&X& X&X& -&-& -&-\\
            phone-10 & X&X& X&X& -&-& -&- \\
            phone-10-long & X&X& X&X& -&-& -&- \\
            phone-10-long-repeat & X&X& X&X& -&-& -&- \\
            phone-10-short & X&X& X&X& -&-& -&- \\
univ-1 & && && -&-& &\\
            univ-1-long & && && -&-& &\\
            univ-1-long-repeat & && && -&-& &\\
            univ-1-short & && && -&-& &\\
            univ-2\\
            univ-2-long\\
            univ-2-long-repeat\\
            univ-2-short\\
            univ-3\\
            univ-3-long\\
            univ-3-long-repeat\\
            univ-3-short\\
            univ-4\\
            univ-4-long\\
            univ-4-long-repeat\\
            univ-4-short\\
            univ-5\\
            univ-5-long\\
            univ-5-long-repeat\\
            univ-5-short\\
            univ-1\\
            univ-6-long\\
            univ-6-long-repeat\\
            univ-6-short\\
        }
 \hline
        \end{tabular}
    \end{center}
    }
    \caption{Runtimes and SFTA sizes for selected SyGuS 2018 benchmarks}
    \label{table:sygus}
\end{figure*}

We evaluate the scalability of our implementation by applying it to all
problems in the SyGuS 2018 benchmark suite~\cite{SyGuS2018}. For each 
problem we use the clean (noise-free) data set for the problem provided with the benchmark suite.
We use the lexicographic objective function $U_L(l,c)$ 
with the $0/\infty$ loss function and the $c = \mathrm{Size}(p)$ complexity measure.
We run each benchmark with bounded scope height threshold  $d =$ 1, 2, 3,
and 4 and record the running time on that benchmark problem
and the number of states in the SFTA.  
A state with symbol $e$ is only added 
to the SFTA if the length of its output value is not greater than the length of the output string.

Because the running time of our technique does not depend on the specific
utility function (except for the time required to evaluate the utility
function, which is typically negligible for most utility functions, and
except for search space pruning techniques appropriate for specific combinations
of utility functions and DSLs), we anticipate that these results will 
generalize to other utility functions. All experiments are run on an 
3.00 GHz Intel(R) Xeon(R) CPU E5-2690 v2 
machine with 512GB memory running Linux 4.15.0. 
With a timeout limit of 10 minutes and bounded scope height threshold of 4, 
the implementation is able to solve 64 out of the 108 SyGuS 2018 benchmark problems. 
For the remaining 48 benchmark problems a correct program does not exist within
the DSL at bounded scope height threshold 4.

Table~\ref{table:sygus} presents results for selected SyGuS 2018 benchmarks.
We omit all name-combine-4-*, 
phone-3-*, phone-4-*, phone-9-*, phone-10-*, and univ-* benchmarks --- 
all runs for these benchmarks either do not synthesize a correct program or do not terminate. 
Table~\ref{table:sygus} presents results for all other SyGuS 2018 benchmarks.
Our synthesis technique removes all duplicates from the dataset in case of
$0/\infty$ loss.
There is a row for each benchmark problem.  
The first column presents the name of the benchmark. The next four columns present
results for the technique running with bounded scope height threshold $d =$ 1, 2, 3,
and 4, respectively. Each column has two subcolumns: the first
presents the running time on that benchmark problem 
(in seconds); the second presents the number of states in the SFTA (in 
thousands of states). An entry X indicates that the implementation
terminated but did not synthesize a correct program that agreed with all
provided input-output examples. An entry - indicates that the implementation
did not terminate. 

\Comment{
The full paper presents the results.\footnote{
Available at https://people.csail.mit.edu/rinard/paper/fse20.pdf
}
With a timeout limit of 10 minutes and bounded scope height threshold of 4, 
the implementation is able to solve 64 out of the 108 benchmark problems. 
For the remaining 48 benchmark problems a correct program does not exist within
the DSL at bounded scope height threshold 4. 
Correct programs for the *-initials-* and many of the name-combine-*
benchmark problems do not exist in the search spaces until
the bounded scope height threshold becomes 3 or more. 
}

In general, both the running times and the number of states in the
SFTA increase as the number of provided input-output 
examples and/or the bounded height threshold increases.  
The SFTA size sometimes stops increasing as the height
threshold increases. We attribute this phenomenon to the
application of a search space pruning technique that terminates
the recursive application of the production $e := \mathsf{Concat}(f, e);$
when the generated string becomes longer than the current  output string ---
in this case any resulting synthesized program will produce an
output that does not match the output in the data set. 

We compare with a previous technique that uses FTAs to solve
program synthesis problems~\cite{wang2017synthesis}. This previous technique 
requires clean data and only synthesizes programs that agree
with all input-output examples in the data set. Our technique builds SFTAs with similar
structure, with additional overhead coming from the evaluation
of the objective function. We obtained the implementation 
of the technique presented in~\cite{wang2017synthesis} and ran this implementation
on all benchmarks in the SyGuS 2018 benchmark suite. The running
times of our implementation and this previous implementation are
comparable. 

\begin{figure*}
    {\tiny
    \begin{center}
        \begin{tabular}{|l|c|c |c|c|c|c|  c|c|c|c|  c|c|c|c|  c|c|c|c| }
            \hline
            \multicolumn{3}{|l|}{Threshold} 
            & \multicolumn{4}{c|}{1} &  \multicolumn{4}{c|}{2} & 
            \multicolumn{4}{c|}{3} & \multicolumn{4}{c|}{4}   \\ 
            \hline
            Benchmark Name & n & output size & time(sec) & SFTA size  
            & DL Loss & size
& time(sec) & SFTA size  
            & DL Loss & size 
& time(sec) & SFTA size  
            & DL Loss & size 
& time(sec) & SFTA size  
            & DL Loss & size 

            \\
            \hline
            \input{dlres}
            \hline
        \end{tabular}
    \end{center}
    }
    \caption{Runtimes, SFTA sizes, Synthesized Program Loss, and its size for
    SyGuS 2018 benchmarks under DL Loss}
    \label{tbl:dlsygus}
\end{figure*}

\begin{figure}
\begin{center}
\begin{tabular}{|l|c|c |c|c|}
\hline
    \multirow{2}{*}{Benchmark} & \multirow{2}{*}{Data Set} & \multicolumn{3}{c|}{Number of Required} \\
    \multirow{2}{*}{} & \multirow{2}{*}{Size} & \multicolumn{3}{c|}{Correct Examples} \\
    \cline{3-5}
    & & 1-Delete&DL & 0/1\\
\hline
    bikes & 6 &0& 0 & 2\\
    dr-name & 4& 0&0 & 1\\
    firstname & 4&0& 0 & 2\\
    lastname & 4& 0& 2 & 4\\ 
    initials & 4& 0& 2 & 2\\
    reverse-name & 6 & 0& 0 & 2\\
    name-combine & 4& 0& 0 & 2\\
    name-combine-2 & 4 & 0& 0 & 2\\
    name-combine-3 & 4 & 0& 0 & 2\\
    phone & 6 & 0& 2 & 3\\
    phone-1 & 6 & 0& 3 & 3\\
    phone-2 & 6& 0& 2 & 3\\
    phone-5 & 7 & 0& 2 & 3\\
    phone-6 & 7 & 0& 2 & 3\\
    phone-7 & 7 & 0& 2 & 3\\
    phone-8 & 7 & 0& 0 & 1\\
\hline
\end{tabular}
\end{center}
\vspace{-.1in}
    \caption{Minimum number of correct examples required to synthesize correct a program.}
    \label{table:minexample}
\vspace{-.2in}
\end{figure}

\vspace{-.1in}
\subsection{Noisy Data Sets, Character Deletion} 
\label{sec:noisycd}

We next present results for our implementation running on small
(few input-output examples) data sets with character deletions. We use a noise
source that cyclically deletes a single character from each output 
in the data set in turn, starting with the first character, proceeding 
through the output positions, then wrapping around to the first character again. 
We apply this noise source to corrupt every output in the data set. 
To construct a noisy data set with $k$ correct (uncorrupted) outputs, we
do not apply the noise source to the last $k$ outputs in the data set. 

We exclude all *-long, *long-repeat, and *-short benchmarks and all benchmarks that
do not terminate within the time limit at height bound 4. For each remaining
benchmark we use our implementation and the generated corrupted data sets to 
determine the minimum number of correct outputs in the 
corrupted data set required for the implementation to produce a correct
program that matches the original hidden clean data set on all input-output
examples. We consider three loss functions: the $0/1$ 
and DL loss functions (Section~\ref{sec:lossFunction}) and the
following 1-Delete loss function, which is designed to work with
noise sources that delete a single character from the output stream:

\begin{definition} {\bf 1-Delete Loss Function:} The 1-Delete loss function 
$\cL_{1D}(p, \cD)$ uses the per-example loss function $L$ that is 
0 if the outputs from the synthesized program and the data set match exactly, 1 if a single
deletion enables the output from the synthesized program to match
the output from the data set, and $\infty$ otherwise: 
\[
    \cL_{1D} (p, \cD) = \sum\limits_{(\sigma_i, o_i) \in \cD} L_{1D}(\llbracket p
    \rrbracket \sigma_i, o_i), \mbox{ where }
\]
\begin{equation*}
    L_{1D}(o_1, o_2) = \begin{cases}
        0 & o_1 = o_2\\
        1 & a \cdot x \cdot b = o_1 \wedge a \cdot b = o_2 \wedge \vert x \vert =
        1\\ 
        \infty & \text{ otherwise}
    \end{cases}
\end{equation*}
\end{definition}

We use the lexicographic objective function $U_L(l,c)$ 
with $c = \mathrm{Size}(p)$ as the complexity measure and bounded scope
height threshold $d = 4$. We apply a search space pruning technique that 
terminates the recursive application of the production $e := \mathsf{Concat}(f, e);$
when the generated string becomes more than one character longer than the 
current  output string. 

Table~\ref{table:minexample} summarizes the results. 
The Data Set Size Column presents the total number of input-output examples in the
corrupted data set. The next three columns, 1-Delete, DL, and 0/1, present the minimum number
of correct (uncorrupted) input-output examples required for 
the technique to synthesize a correct program (that agrees with the 
original hidden clean data set on all input-output examples) using 
the 1-Delete, DL, and 0/1 loss functions, respectively. 

With the 1-Delete loss function, the minimum number of required correct
input-output examples is always 0 --- the implementation synthesizes,
for every benchmark problem, a correct program that matches every input-output 
example in the original clean data set even when given a data set in which
every output is corrupted. This result highlights how 1) discrete noise sources 
produce noisy outputs that leave a significant amount of information from the original
uncorrupted output available in the corrupted output and 2) a loss function
that matches the noise source can enable the synthesis technique to 
exploit this information to produce correct programs even in the 
face of substantial noise.

With the DL loss function, the implementation synthesizes a correct
program for 8 of the 16 benchmarks when all outputs in the data set 
are corrupted. For 7 of the remaining 8 benchmarks
the technique requires 2 correct input-output examples to synthesize
the correct program. The remaining benchmark requires 3 correct examples. 
The general pattern is that the technique tends to require correct examples
when the output strings are short. The synthesized incorrect programs typically 
use less of the input string. 

These results highlight how the DL loss function still extracts 
significant useful information available in outputs corrupted
with discrete noise sources. But in comparison with the 1-Delete loss function,
the DL loss function is not as good a match for the character deletion noise
source. The result is that the synthesis technique, when working with the 
DL loss function, works better with longer inputs, sometimes encounters incorrect
programs that fit the corrupted data better, and therefore sometimes requires
correct inputs to synthesize the correct program. 

With the 0/1 loss function, the technique always requires at least one and usually more correct
inputs to synthesize the correct program. In contrast to the 1-Delete and DL 
loss functions, the 0/1 loss function does not extract information from corrupted outputs. 
To synthesize a correct program with the 0/1 loss function in this scenario, the
technique must effectively
ignore the corrupted outputs to synthesize the program working only with information
from the correct outputs. It therefore always requires at least one and usually more
correct outputs before it can synthesize the correct program.

\vspace{-.1in}
\subsection{Noisy Data Sets, Character Replacements}

We next present results for our implementation running on larger data
sets with character replacements. The phone-*-long-repeat benchmarks
within the SyGuS 2018 benchmarks contain transformations over phone
numbers. The data sets for these benchmarks contain 400 input-output
examples, including repeated input-output examples. 

For each of these phone-*-long-repeat benchmark problems on which our technique terminates 
with bounded scope height threshold 4 (Section~\ref{sec:scalability}), we construct
a noisy data set as follows. For each digit in each output string in the
data set, we flip a  biased coin. With probability $b$, we replace the
digit with a uniformly chosen random digit 
(so that each digit in the noisy output is not equal to the original 
digit in the uncorrupted output with probability $9/10\times b$).

We then run our implementation on each benchmark problem with the 
noisy data set using the tradeoff objective function $U_\lambda(l,c) = l + \lambda \times c$
with complexity measure $c = \mathrm{Size}(p)$ and the 
following $n$-Substitution loss function:

\begin{definition}
{\bf $n$-Substitution Loss Function:} \\ The $n$-Substitution loss function 
$\cL_{nS}(p, \cD)$ uses the per-example loss function $L_{nS}$ that 
counts the number of positions where the noisy output
does not agree with the output from the synthesized program.  If the synthesized
program produces an output that is longer or shorter than the output in the 
noisy data set, the loss function is $\infty$:
\[
    \cL_{nS} (p, \cD) = \sum\limits_{(\sigma_i, o_i) \in \cD} L_{nS}(\llbracket p
    \rrbracket \sigma_i, o_i), \mbox{ where }
\]
\begin{equation*}
    L_{nS}(o_1, o_2) = \begin{cases}
        \infty & \vert o_1 \vert \neq \vert o_2 \vert \\
        \sum\limits_{i=1}^{\vert o_1 \vert} 1 \mbox{ if } o_1[i] \neq o_2[i] \mbox{ else } 0
    & \vert o_1 \vert = \vert o_2 \vert 
    \end{cases}
\end{equation*}
\end{definition}

We run the implementation for all combinations of the bounded scope threshold $b \in \{ 0.2, 0.4, 0.6 \}$ and
$\lambda \in \{ 0.001, 0.1 \}$. For every combination of $b$ and $\lambda$, 
and for every one of the phone-*-long-repeat benchmarks in the SyGuS
2018 benchmark set, the implementation
synthesizes a correct program that produces the same outputs as in the
original (hidden) clean data set. 

These results highlight, once again, the ability of our technique to work 
with loss functions that match the characteristics of discrete noise sources
to synthesize correct programs even in the face of substantial noise. 

\begin{figure}[t]
\begin{center}
\begin{tabular}{|l|c|c |c|c|}
\hline
Benchmark & Data set& DL & Output & Program  \\
    & Size & Loss & Size & Size\\
    \hline
name-combine-4 & 5 & 10 & 49 & 16 \\
phone-3 & 7 & 14 & 91 & 11\\
phone-4 & 6 & 6 & 66 & 17\\
phone-9 & 7 & 14 & 99 & 21\\
phone-10 & 7 & 14 & 120 & 21\\
\hline
\end{tabular}
\end{center}
\vspace{-.1in}
    \caption{Approximate program synthesis with DL loss.}
    \label{tbl:approximate}
\vspace{-.2in}
\end{figure}

\vspace{-.1in}
\subsection{Approximate Program Synthesis}

\Comment{
As observed in Table~\ref{tbl:sygus}, for some benchmarks, a correct program
does not exist within the DSL at low bounded scope threshold. We use our
synthesis technique to find the 
program which best fits the entire dataset given the
{\it Damerau-Levenshtein} Loss Function.
Table~\ref{tbl:dlsygus}
presents results from our implementation on the clean (noise-free) benchmark data sets
with the DL loss function, $\mathsf{Size}(p)$
complexity measure, and 
lexicographic objective function $U_L(\cL_{DL}(p,\cD), \mathrm{Size}(p))$. 
The first column presents the name of the benchmark. The next column presents
the number of input-output examples in the given benchmark.  
The next four columns present
results for the technique running with bounded scope height threshold $d =$ 1, 2, 3,
and 4, respectively. Each column has four subcolumns: the first
presents the running time on that benchmark problem 
(in seconds). The second presents the number of states in the SFTA (in 
thousands of states). The third presents the DL loss of the synthesized program
over the entire dataset. The fourth presents the size of the synthesized program. 
 An entry - indicates that the implementation
did not terminate. 
}

For the benchmarks in Table~\ref{tbl:approximate}, a correct program does
not exist within the DSL at bounded scope threshold 2. Table~\ref{tbl:approximate}
presents results from our implementation on the clean (noise-free) benchmark data sets
with the DL loss function, $\mathsf{Size}(p)$
complexity measure, lexicographic objective function $U_L(\cL_{DL}(p,\cD), \mathrm{Size}(p))$,
and bounded scope threshold 2. 
The first column presents the name of the benchmark. The next four columns
present the number of input-output examples in the benchmark data set, the DL loss incurred by the
synthesized program over the entire data set, the sum of the lengths of the
output strings of the data set (the DL loss for an empty output would be this sum), 
and the size of the synthesized program.

For the phone-* benchmarks, a correct program 
outputs the entire input telephone number but changes the punctutation, for example
by including an area code in parentheses. 
The synthesized approximate programs correctly preserve the telephone number 
but apply only some of the punctuation changes. The result is $2 = 14/7$ characters
incorrect per output for all but phone-4, which has 1 character per output incorrect. 
Each output is between $13=91/7$ and $17=120/7$ characters long. 
For name-combine-4, the synthesized approximate program
correctly extracts the last name, inserts a comma and a period, but does
not extract the initial of the first name. These results highlight the ability
of our technique to approximate a correct program when the correct program
does not exist in the program search space. 

Table~\ref{tbl:dlsygus}
presents results from our implementation on the clean (noise-free) benchmark data sets
in SyGuS 2018.
The first column presents the name of the benchmark. The next column presents
the number of input-output examples in the given benchmark. The next column 
presents 
the sum of the lengths of the
output strings of the data set.
The next four columns present
results for the technique running with bounded scope height threshold $d =$ 1, 2, 3,
and 4, respectively. Each column has four subcolumns: the first
presents the running time on that benchmark problem 
(in seconds). The second presents the number of states in the SFTA (in 
thousands of states). The third presents the DL loss of the synthesized program
over the entire dataset (compare this DL loss with the sum of the output lengths
over the data set). The fourth presents the size of the synthesized program. 
 An entry - indicates that the implementation
did not terminate.

\vspace{-.1in}
\subsection{Discussion}

\noindent{\bf Practical Applicabilty:} 
The experimental results show that our technique is effective
at solving string manipulation program synthesis problems with modestly sized
solutions like those present in the SyGuS 2018 benchmarks. More specifically, 
the results highlight how the combination of structure from the DSL, a discrete noise source that
preserves some information even in corrupted outputs, and a good match
between the loss function and noise source can enable very effective synthesis
for data data sets with only a handful of input-output examples
even in the presence of substantial noise.
Even with as generic a loss function as the 0/1 loss function, the technique is effective at 
dealing with data sets in which a significant fraction of the outputs are corrupted. 
We anticipate that these results will generalize to similar classes of program synthesis problems
with modestly sized solutions within a tractable and focused class of computations. 

We note that our current implementation does not scale to SyGuS 2018 benchmarks with larger
solutions. These benchmarks were designed to test the scalability of current and future program 
synthesis systems.  No currently extant program analysis system of which we are aware can solve 
these larger problems. 

To the extent that the SyGuS 2018 bencharks accurately represent the kinds of program
synthesis problems that will be encountered in practice, our results provide encouraging
evidence that our technique can help program synthesis systems work effectively with noisy data sets. 
Important future work in this area will more fully investigate interactions between
the DSL, the noise source, the loss function, the classes of synthesis problems that
occur in practice, and the scalability of the synthesis technique. A full evaluation
of the immediate practical applicability of program synthesis for noisy data sets, as well as
a meaningful evaluation of program synthesis more generally, awaits this future work. 

\noindent{\bf Noise Sources With Different Characteristics:} 
Our experiments largely consider discrete noise sources that preserve some information in 
corrupted outputs. The results highlight how loss functions like the 1-Delete, DL,
and $n$-Substitution loss functions 
can enable our technique to extract and exploit this preserved information to enhance
the effectiveness of the synthesis. The question may arise how well may our technique perform 
with noise sources that leave little or even no information intact in corrupted outputs?
Here the results from the $0/1$ loss function, which does not aspire to extract any 
information from corrupted inputs, may be relevant --- if the corrupted outputs considered
together do not conform to a target computation in the DSL, the technique will, in effect,
ignore these corrupted outputs to synthesize the program based on any remaining uncorrupted outputs.
A final possibility is that the noise source may systematically produce outputs characteristic 
of a valid but incorrect computation. Here we would expect the algorithm to require 
a balance of correct outputs before it would be able to synthesize the correct program.

\vspace{-.1in}
\section{Related Work}
\Comment{
We compare our technique against related approaches to solving similar problems.
}

The problem of learning programs from a set of input-output 
examples has been studied extensively~\cite{shaw1975inferring, 
gulwani2011automating, singh2016transforming}.
These techniques can be largely broken down into the following four categories:


\noindent
        {\bf Synthesis Using Solvers:} 
These systems require the user 
to provide precise semantics for the operators for the DSL they are using 
\cite{itzhaky2010simple}. Our technique, in contrast, 
only requires black-box implementations 
of these operators. A large
class of these systems depend on solvers which do not 
scale as the number of examples increases.
Since our techniques are based on tree automata, our cost linearly 
increases as the number of examples increase. 
These systems require all input-output examples 
to be correct and only synthesize programs that match all
input-output examples. 

\noindent
        {\bf Enumerative Techniques:} 
These techniques search the 
space of programs to find a single program that is consistent with 
the given examples~\cite{feser2015synthesizing, osera2015type}. 
Specifically, they enumerate all programs in 
the given DSL and terminate when they find the correct program. 
These techniques may apply different heuristics/techniques to 
prune the search space/speed up this process~\cite{osera2015type}.
These techniques require all input-output examples
to be correct and only synthesize programs that match all
input-output examples. 

\noindent
        {\bf VSA-based/Tree Automata-based Techniques:} These \\techniques
build complex data structures representing all possible programs 
compatible with the given examples~\cite{singh2016transforming,
polozov2015flashmeta, wang2017synthesis}. Current work requires
users to provide correct input-output examples. 
Our work modifies these techniques to handle noisy data and
to synthesize approximate programs that minimize an objective
function over the provided potentially noisy data set. 

\noindent
    {\bf Neural Program Synthesis/ML Approaches:}
There is extensive work that uses 
machine learning/deep neural networks to synthesize programs~\cite{raychev2016learning, 
devlin2017robustfill, balog2016deepcoder}. These techniques require a training phase
and a differentable loss function. Our technique requires
no training phase and can work with arbitrary loss functions including, for example, the 
$0/1$ loss function. Machine learning techniques are incompatible with this type of loss function. 
These systems provide no guarantees over the completeness and the optimality 
of their result, whereas our technique, due to its property of exploring all
programs of size less than a threshold, always finds a program within the 
bounded scope that minimizes the objective function.


\noindent{\bf Data Set Sampling or Cleaning:} 
There has been recent work which aspires to clean the data set or pick
representative examples from the data set for
synthesis~\cite{gulwani2011automating, raychev2016learning, pu2017selecting},
for example by using machine learning or data cleaning to select productive
subsets of the data set over which to perform exact synthesis. 
In contrast to these techniques, our proposed techniques 1) provide
deterministic guarantees (as opposed to either probabilistic guarantees
as in~\cite{raychev2016learning} or no guarantees at all as
in~\cite{pu2017selecting, gulwani2011automating}), 2) do not
require the use of oracles as in~\cite{raychev2016learning}, 3) can operate 
successfully even on data sets
in which most or even all of the input-output examples are corrupted,
and 4) do not require the explicit selection of a subset of the
data set to drive the synthesis as in~\cite{gulwani2011automating,
raychev2016learning}. 

\noindent{\bf Active Learning:} Recent research exploits the availability of
a reference implementation to use active learning for program synthesis~\cite{shen2019using}. Our
technique, in contrast, works directly from given input-output examples with no
reference implementation. 

\section{Conclusion}

Dealing with noisy data is a pervasive problem in modern computing
environments. Previous program synthesis systems target data sets
in which all input-output examples are correct to synthesize
programs that match all input-output examples in the data set. 

We present a new program synthesis technique for working with noisy data
and/or performing approximate program synthesis. 
Using state-weighted finite tree automata, this technique supports the
formulation and solution of a variety of new program synthesis problems
involving noisy data and/or approximate program synthesis. 
The results highlight how this technique, by exploiting information from a variety
of sources --- structure from the underlying DSL, information left intact
by discrete noise sources --- can deliver effective program synthesis even
in the presence of substantial noise.

\vspace{-.1in}
\begin{acks}                            
This research was supported, in part, by the Boeing Corporation and DARPA
Grants N6600120C4025 and HR001120C0015. 
\end{acks}


\arxiv{

\clearpage
\newpage }
\bibliography{citation}
\flushend


\end{document}